\documentclass[envcountsame]{llncs}
    \usepackage{amsmath}
    \usepackage{amsxtra}
    \usepackage{amstext}
    \usepackage{amssymb}
    \usepackage{latexsym}
    \usepackage{dsfont} 
\newcommand{\tr}[2][1]{Tr_{#1}^{#2}}

\newcommand{\GF}[1]{{\mathbb F}_{#1}}

\begin{document}
\title{Solving $X^{q+1}+X+a=0$ over Finite Fields}
\author{Kwang Ho Kim\inst{1,2}\and Junyop Choe\inst{1} \and Sihem Mesnager\inst{3}}

\institute{ Institute of Mathematics, State Academy of Sciences,
Pyongyang, Democratic People's Republic of Korea\\
\email{khk.cryptech@gmail.com} \and PGItech Corp., Pyongyang, Democratic People's Republic of Korea\\ \and Department of Mathematics, University of Paris VIII, 93526 Saint-Denis, France,  LAGA UMR 7539,  CNRS, Sorbonne Paris
Cit\'e, University of Paris XIII, 93430 Villetaneuse, France and Telecom ParisTech, 91120  Palaiseau, France\\
\email{smesnager@univ-paris8.fr}\\} \maketitle

\begin{abstract}

Solving the equation $P_a(X):=X^{q+1}+X+a=0$ over finite field
$\GF{Q}$, where $Q=p^n, q=p^k$ and $p$ is a prime, arises in many
different contexts including finite geometry, the inverse Galois
problem \cite{ACZ2000}, the construction of difference sets with
Singer parameters \cite{DD2004}, determining cross-correlation
between $m$-sequences \cite{DOBBERTIN2006,HELLESETH2008} and to
construct error-correcting codes \cite{Bracken2009}, as well as to
speed up the index calculus method for computing discrete logarithms
on finite fields \cite{GGGZ2013,GGGZ2013+} and on algebraic curves
\cite{M2014}.

 Subsequently, in
\cite{Bluher2004,HK2008,HK2010,BTT2014,Bluher2016,KM2019,CMPZ2019,MS2019},
the $\GF{Q}$-zeros  of $P_a(X)$ have been studied: in
\cite{Bluher2004} it was shown that the possible values of the
number of
 the zeros that $P_a(X)$ has in $\GF{Q}$ is $0$, $1$, $2$ or $p^{\gcd(n, k)}+1$.
 Some criteria for the number of the $\GF{Q}$-zeros of $P_a(x)$  were
 found in \cite{HK2008,HK2010,BTT2014,KM2019,MS2019}.
However, while the ultimate goal is to identify all the
$\GF{Q}$-zeros,
 even in the case $p=2$, it was solved only under the condition $\gcd(n, k)=1$ \cite{KM2019}.

We discuss this equation without any restriction on $p$ and
$\gcd(n,k)$.  New criteria for the number of the $\GF{Q}$-zeros of
$P_a(x)$ are proved. For the cases of one or two $\GF{Q}$-zeros, we
provide explicit expressions for these rational zeros in terms of
$a$. For the case of $p^{\gcd(n, k)}+1$ rational zeros, we provide a
parametrization of such $a$'s and  express the $p^{\gcd(n, k)}+1$
rational zeros by using that parametrization.

\noindent\textbf{Keywords:} Equation $\cdot$
M\"{u}ller-Cohen-Matthews (MCM) polynomial $\cdot$ Dickson
polynomial $\cdot$ Zeros of a polynomial $\cdot$ Irreducible
polynomial.
\end{abstract}

\section{Introduction}
Let $k$ and $n$ be any positive integers with $\gcd(n,k)=d$. Let
$Q=p^n$ and $q=p^k$ where $p$ is a prime. We consider the polynomial
\[P_a(X):=X^{q+1}+X+a, a\in \GF{Q}^{*}.\]  Notice the more general polynomial forms
$X^{q+1}+rX^{q}+sX+t$ with $s\neq r^q$ and $t\neq rs$ can be
transformed into this form by the substitution
 $X=(s-r^q)^{\frac{1}{q}}X_1-r$. It is clear that $P_a(X)$ have no multiple roots.

 These polynomials have arisen in several different contexts including finite geometry, the inverse Galois
problem \cite{ACZ2000}, the construction of difference sets with
Singer parameters \cite{DD2004}, determining cross-correlation
between $m$-sequences \cite{DOBBERTIN2006,HELLESETH2008} and to
construct error-correcting codes \cite{Bracken2009}. These
polynomials are also exploited to speed up (the relation generation
phase in) the index calculus method for computation of discrete
logarithms on finite fields \cite{GGGZ2013,GGGZ2013+} and on
algebraic curves \cite{M2014}.

Let $N_a$ denote the number of zeros in $\GF{Q}$ of polynomial
$P_a(X)$ and $M_i$ denote the number of $a\in \GF{Q}^{*}$ such that
$P_a(X)$ has exactly $i$ zeros in $\GF{Q}$.  In 2004, Bluher
\cite{Bluher2004}  proved that $N_a$ takes either of 0, 1, 2 and
$p^d+1$ where $d=\gcd(k, n)$ and computed $M_i$ for every $i$. She
also stated some criteria for the number of the $\GF{Q}$-zeros of
$P_a(X)$.

The ultimate goal in this direction of research is to identify all
the $\GF{Q}$-zeros of $P_a(X)$. Subsequently, there were much
efforts for this goal, specifically for a particular instance of the
problem over binary fields i.e. $p=2$. In 2008 and 2010, Helleseth
and Kholosha \cite{HK2008,HK2010} found new criteria for the number
of $\GF{2^n}$-zeros of $P_a(X)$. In the cases when there is a unique
zero or exactly two zeros and $d$ is odd, they provided explicit
expressions of these zeros as polynomials of $a$ \cite{HK2010}. In
2014, Bracken, Tan and Tan \cite{BTT2014} presented a criterion for
$N_a=0$  in $\GF{2^n}$ when $d=1$ and $n$ is even. Very recently,
Kim and Mesnager \cite{KM2019} completely solved this equation
$X^{2^k+1}+X+a=0$ over $\GF{2^n}$ when $d=1$. They showed that the
problem of finding zeros in $\GF{2^n}$ of $P_a(X)$ in fact can be
divided into two problems with odd $k$: to find the unique preimage
of an element in $\GF{2^n}$ under a MCM polynomial and to find
preimages of an element in $\GF{2^n}$ under a Dickson polynomial. By
completely solving these two independent problems, they explicitly
calculated all possible zeros in $\GF{2^n}$ of $P_a(X)$, with new
criteria for which $N_a$ is equal to $0$, $1$ or $p^d+1$ as a
by-product.

Very recently, new criteria for which $P_a(X)$ has $0$, $1$, $2$ or
$p^d+1$ roots were stated by \cite{MS2019} for any characteristic.

We discuss the equation $X^{p^k+1}+X+a=0, a\in \GF{p^n}$, without
any restriction on $p$ and $\gcd(n,k)$.  After defining a sequence
of polynomials and considering its properties in Section 2, it is
shown in Section 3 that if $N_a\leq2$ then there exists a quadratic
equation that the rational zeros must satisfy. In Section 4, we
state some useful properties of the polynomials which appear as the
coefficients of that quadratic equation. In Section 5, new criteria
for the number of the $\GF{Q}$-zeros of $P_a(x)$ are proved. For the
cases of one or two $\GF{Q}$-zeros, we provide explicit expressions
for these rational zeros in terms of $a$. We also provide a
parametrization of the $a$'s for which $P_a(X)$ has $p^{\gcd(n,
k)}+1$ rational zeros. Based that parametrization, all the
$p^{\gcd(n, k)}+1$ rational zeros are also expressed. For the case
of $p^{\gcd(n, k)}+1$ rational zeros,  some results to explicitly
express these rational zeros in terms of $a$ are further presented
in Section 6. Finally, we conclude in Section 7.

\section{Preliminaries}

 Given positive integers $k$ and $l$, define a polynomial
\[T^{kl}_{k}(X):=X+X^{p^k}+\cdots+X^{p^{k(l-2)}}+X^{p^{k(l-1)}}.\]
Usually we will abbreviate $T^{l}_{1}(\cdot)$ as $T_l(\cdot)$. For
$x\in \GF{p^l}$, $T_l(x)$ is the absolute trace $\tr{l}(x)$ of $x$.
The zeros of this polynomial are studied in \cite{KCLGM2019}.  In
particular, we need the following.
\begin{proposition}\label{zerosTk} For any positive integers $k$ and
$r$,
\[
\{x\in \overline{\GF{p}} \mid T^{kr}_{k}(x)=0\}=\{ u- u^{p^k}\mid
u\in \GF{p^{kr}}\}.
\]
\end{proposition}
\begin{proof}
Evidently, $\{u-u^{p^k}|u\in\GF{p^{kr}}\}\subset
\mathbf{ker}(T_k^{kr})$.  The linear mapping $u\mapsto u-u^{p^k}$
has the kernel $\GF{p^k}$ and so
$\#\{u-u^{p^k}|u\in\GF{p^{kr}}\}=p^{k(r-1)}$. On the other hand,
$T_k^{kr}$ can not have a kernel of  greater cardinality than its
degree $p^{k(r-1)}$.\qed
\end{proof}

Define the sequence of polynomials $\{A_r(X)\}$ as follows:
\begin{equation}\label{eq.defA}\begin{aligned}
&A_1(X)=1, A_2(X)=-1,\\
&A_{r+2}(X)=-A_{r+1}(X)^{q}-X^{q}A_{r}(X)^{q^2} \text{ for } r\geq
1.
\end{aligned}\end{equation}

The following lemma gives another identity which can be used as an
alternative definition of $\{A_r(X)\}$ and an interesting property
of this polynomial sequence which will be importantly applied
afterwards.
\begin{lemma}\label{Lem.newdef}
 For any $r\geq 1$, the following are true.
\begin{enumerate}
\item \begin{equation}\label{eq.defAA}
A_{r+2}(X)=-A_{r+1}(X)-X^{q^r}A_{r}(X).
\end{equation}
\item \begin{equation}\label{eq.Norm}
A_{r+1}(X)^{q+1}-A_r(X)^qA_{r+2}(X)=X^{\frac{q(q^r-1)}{q-1}}.
\end{equation}
\end{enumerate}
\end{lemma}
\begin{proof}
We will prove these identities by induction on $r$. It is easy to
check that they hold for $r=1,2$. Suppose that they hold for all
indices less than $r (\geq 3)$. Then, we have
\[\begin{aligned}
A_{r+3}(X)&=-A_{r+2}(X)^{q}-X^{q}A_{r+1}(X)^{q^2} \\
&=\left(A_{r+1}(X)+X^{q^{r}}A_{r}(X)\right)^{q} +
X^{q}\left(A_{r}(X)+X^{q^{r-1}}A_{r-1}(X)\right)^{q^2}\\
&=\left(A_{r+1}^{q}(X)+X^qA_{r}^{q^2}(X)\right) +
X^{q^{r+1}}\left(A_{r}^{q}(X)+X^qA_{r-1}^{q^2}(X)\right)\\
&=-A_{r+2}(X)-X^{q^{r+1}}A_{r+1}(X),
\end{aligned}\]
which proves \eqref{eq.defAA} for all $r$. Also, using the proved
equality \eqref{eq.defAA}, we have
\[\begin{aligned}
&A_{r+2}(X)^{q+1}-A_{r+1}(X)^qA_{r+3}(X)\\
&=A_{r+2}(X)^{q+1}+A_{r+1}(X)^q\left(A_{r+2}(X)+X^{q^{r+1}}A_{r+1}(X)\right)\\
&=X^{q^{r+1}}\left(A_{r+1}(X)^{q+1}-A_r(X)^qA_{r+2}(X)\right)+A_{r+2}(X)\left(A_{r+2}(X)^q+A_{r+1}(X)^q+X^{q^{r+1}}A_{r}(X)^q\right)\\
&\overset{\eqref{Lem.newdef}}{=}X^{q^{r+1}}\left(A_{r+1}(X)^{q+1}-A_r(X)^qA_{r+2}(X)\right)\\
&=X^{q^{r+1}}X^{\frac{q(q^r-1)}{q-1}}=X^{\frac{q(q^{r+1}-1)}{q-1}},
\end{aligned}\]
which proves \eqref{eq.Norm} for all $r$.\qed
\end{proof}

The zero set of $A_r(X)$ can be completely determined for all $r$:
\begin{proposition}\label{Prop.A_r(x)=0}
For any $r\geq 3$,
\[
\{x\in \overline{\GF{p}} \mid
A_r(x)=0\}=\left\{\frac{(u-u^q)^{q^2+1}}{(u-u^{q^2})^{q+1}}, \ \
u\in \GF{q^r}\setminus\GF{q^2}\right\}.
\]
\end{proposition}
\textbf{Proof.} Given any $x\in \overline{\GF{p}}\setminus \{0\}$,
there exists at least one element $v\in \overline{\GF{p}}$ such that
$x = \frac{v^{q^2+1}}{(v+v^q)^{q+1}}$ and $v+v^q\neq0$. Then, for
any $r\geq 2$, we have
\[
A_r(x)=(-1)^{r+1}\frac{\sum_{j=1}^{r}v^{q^j}}{v^q+v^{q^2}}\prod_{j=2}^{r-1}
\left(\frac{v}{v+v^q}\right)^{q^j},
\]
where for $i = 2$ it is assumed that the product over the empty set
is equal to 1. Indeed, this can be proved by induction on $r$ as
follows. For $r=2$ and $r=3$, we have
\[\begin{aligned}
A_2(x)=-1=(-1)^3\frac{\sum_{j=1}^{2}v^{q^j}}{v^q+v^{q^2}}
\end{aligned}\]
and \[\begin{aligned}
A_3(x)=1-x^{q}=1-\frac{v^{q+q^3}}{(v+v^q)^{q+q^2}}=(-1)^4\frac{\sum_{j=1}^{3}v^{q^j}}{v^q+v^{q^2}}\left(\frac{v}{v+v^q}\right)^{q^2}.
\end{aligned}\]
Assuming this identity holds for all indices less than $r$, we have
\[\begin{aligned}
A_r(x)&\overset{\eqref{eq.defAA}}{=}-A_{r-1}(x)-x^{q^{r-2}}A_{r-2}(x)\\
&=(-1)^{r+1}\frac{\sum_{j=1}^{r-1}v^{q^j}}{v^q+v^{q^2}}\prod_{j=2}^{r-2}
\left(\frac{v}{v+v^q}\right)^{q^j}-(-1)^{r+1}\frac{v^{q^r}\sum_{j=1}^{r-2}v^{q^j}}{(v+v^q)^{q^{r-1}+q}}\prod_{j=2}^{r-2}
\left(\frac{v}{v+v^q}\right)^{q^j}\\
&=(-1)^{r+1}\frac{(v+v^{q})^{q^{r-1}}\sum_{j=1}^{r-1}v^{q^j}-v^{q^r}\sum_{j=1}^{r-2}v^{q^j}}{v^{q^{r-1}}(v+v^{q})^q}\prod_{j=2}^{r-1}
\left(\frac{v}{v+v^q}\right)^{q^j}\\
&=(-1)^{r+1}\frac{\sum_{j=1}^{r}v^{q^j}}{v^q+v^{q^2}}\prod_{j=2}^{r-1}
\left(\frac{v}{v+v^q}\right)^{q^j}.
\end{aligned}\]

 Thus
$A_r(x)=0$ if and only if
$\sum_{j=1}^{r}v^{q^j}=(T_{k}^{kr}(v))^q=0$ and $v+v^q\neq0$, which
by Proposition \ref{zerosTk} is equivalent to $v=u-u^q$ for some
$u\in \GF{q^r}\setminus\GF{q^2}$.

Therefore, $A_r(x)=0$ if and only if
$x=\frac{(u-u^q)^{q^2+1}}{(u-u^{q^2})^{q+1}}$ for some $u\in
\GF{q^r}\setminus\GF{q^2}$. \qed

\section{Quadratic equation satisfied by rational zeros of $P_a(X)$}

Letting $m=n/d$, define polynomials
\[\begin{aligned}
&F(X):=A_m(X),\\
&G(X):=-A_{m+1}(X)-XA_{m-1}^q(X).
\end{aligned}\]

 We will show that if $F(a)\neq0$ then the $\GF{Q}$-zeros of $P_a(X)$
satisfy a quadratic equation and therefore necessarily $N_a\leq 2$.
\begin{lemma}\label{lem.quadEq}
Let $a\in \GF{Q}^{*}$. If $P_a(x)=0$ for $x\in \GF{Q}$ then
\begin{equation}\label{eq.quadEq}
F(a)x^2+G(a)x+aF^q(a)=0.
\end{equation}
\end{lemma}
\begin{proof} If $x^{q+1}+x+a=0$ for  $x\in \GF{Q}$, then $x\neq 0$ and thus we get
\begin{equation}\label{eq.xq=x+a/x}
x^q = \frac{-x-a}{x}.
\end{equation}
Now, we prove that for any $r\geq 1$
\begin{equation}\label{eq.xqr}
x^{q^r}\left(A_r(a)x-aA_{r-1}(a)^q\right)-A_{r+1}(a)x+aA_{r}(a)^{q}=0
\end{equation}
with the assumption $A_0(x)=0$. In fact, if $r=1$ then the left side
of \eqref{eq.xqr} is $P_a(x)$ and so it holds for $r=1$. Suppose
that it holds for $r\geq 1$. Raising $q-$th power to both sides of
\eqref{eq.xqr} and substituting \eqref{eq.xq=x+a/x}, we have

\begin{center}
$x^{q^{r+1}}\left(A_r(a)^{q}x^{q}-a^{q}A_{r-1}(a)^{q^2}\right)-A_{r+1}(a)^{q}x^{q}+a^{q}A_{r}(a)^{q^2}=0 \Rightarrow$\\
$x^{q^{r+1}}\left(-A_r(a)^{q}\frac{x+a}{x}-a^{q}A_{r-1}(a)^{q^2}\right)+A_{r+1}(a)^{q}\frac{x+a}{x}+a^{q}A_{r}(a)^{q^2}=0\Rightarrow$\\
$x^{q^{r+1}}\left(\big(-A_r(a)^{q}-a^{q}A_{r-1}(a)^{q^2}\big)x-aA_r(a)^q\right)+\left(A_{r+1}(a)^{q}+a^{q}A_{r}(a)^{q^2}\right)x +aA_{r+1}(a)^q=0\Rightarrow$\\
$x^{q^{r+1}}\left(A_{r+1}(a)x-aA_r(a)^q\right)-A_{r+2}(a)x+aA_{r+1}(a)^q=0.$
\end{center}
This shows that \eqref{eq.xqr} holds for $r+1$ and so for all $r$.

Taking $r=m$ in \eqref{eq.xqr} and using the fact that
$x^{q^m}=x^{Q^{k/d}}=x$ when $x\in \GF{Q}$, we obtain the result of
the lemma.\qed
\end{proof}

\section{Some equalities involving $F$ and $G$}
To determine the $\GF{Q}-$rational zeros of $P_a(X)$ when $N_a\leq
2$, we will need the following properties of the polynomials $F$ and
$G$ which appear as coefficients of the quadratic equation
\eqref{eq.quadEq}.
\begin{proposition}\label{Prop.properties_FG} For any $x\in \GF{q^m}$,
the following are true.
\begin{enumerate}
\item
\begin{equation}\label{property_FG_1}
\left(G(x)-2F(x)\right)^q=-G(x).
\end{equation}
\item
\begin{equation}\label{property_FG_2}
{G(x)}^2-4x{F(x)}^{q+1}\in \GF{q}.
\end{equation}
\item
\begin{equation}\label{property_FG_3}
G(x)=-x^qF^{q^2}(x)+F^q(x)+xF(x).
\end{equation}
\end{enumerate}
\end{proposition}
\begin{proof} The first item follows from
\[\begin{aligned}
\left(G(x)-2F(x)\right)^q&={G(x)}^q-2{F(x)}^q=-{A_{m+1}(x)}^q-x^q{A_{m-1}(x)}^{q^2}-2A_m(x)^q\\
&\overset{\eqref{eq.defAA}}{=}(A_{m}(x)+x^{q^{m-1}}A_{m-1}(x))^q-x^q{A_{m-1}(x)}^{q^2}-2A_m(x)^q\\
&=xA_{m-1}(x)^q-x^q{A_{m-1}(x)}^{q^2}-A_m(x)^q \text{  (since $x^{q^m}=x$)}\\
&\overset{\eqref{eq.defA}}{=}xA_{m-1}(x)^q+A_{m+1}(x)=-G(x).
\end{aligned}\]
The second item is proved as follows. Let
$E={G(x)}^2-4x{F(x)}^{q+1}$. Then
\[\begin{aligned}
E^q-E=\left({A_{m+1}(x)}^q+x^q{A_{m-1}(x)}^{q^2}\right)^{2}&-4x^q{A_m(x)}^{q(q+1)}\\
&-\left({A_{m+1}(x)}+x{A_{m-1}(x)}^{q}\right)^{2}+4x{A_m(x)}^{q+1}.
\end{aligned}\]
Consider
$A_{m+1}(x)^q\overset{\eqref{eq.defAA}}{=}(-A_m(x)-x^{q^{m-1}}A_{m-1}(x))^q=-A_m(x)^q-xA_{m-1}(x)^q$.
By substituting this and using \eqref{eq.defA}, we have
\[\begin{aligned}
E^q-E&=\left(-A_m(x)^q-xA_{m-1}(x)^q+x^q{A_{m-1}(x)}^{q^2}\right)^{2}-4x^q{A_m(x)}^{q(q+1)}\\
&-\left(-A_m(x)^q-x^q{A_{m-1}(x)}^{q^2}+x{A_{m-1}(x)}^{q}\right)^{2}+4x{A_m(x)}^{q+1}\\
&=4A_m(x)^q\left(xA_{m-1}(x)^q-x^q{A_{m-1}(x)}^{q^2}\right)-4x^q{A_m(x)}^{q(q+1)}+4x{A_m(x)}^{q+1}\\
&=4A_m(x)^q\left(xA_{m-1}(x)^q-x^q{A_{m-1}(x)}^{q^2}-x^q{A_m(x)}^{q^2}+xA_m(x)\right).
\end{aligned}\]
By the way, since
\[\begin{aligned}
x^q{A_{m-1}(x)}^{q^2}+x^q{A_m(x)}^{q^2}&=x^q(A_{m-1}(x)+A_{m}(x))^{q^2}\\
&\overset{\eqref{eq.defAA}}{=}-x^q(x^{q^{m-2}}A_{m-2}(x))^{q^2}=-x^{q+1}A_{m-2}(x)^{q^2},
\end{aligned}\]
we get
$E^q-E=4xA_m(x)^q\left(A_{m-1}(x)^q+x^{q}A_{m-2}(x)^{q^2}+A_m(x)\right)\overset{\eqref{eq.defA}}{=}0,$
that is, $E={G(x)}^2-4x{F(x)}^{q+1}\in \GF{q}.$

Finally, the third item is verified as follows:
\[\begin{aligned}
G(x)&=-A_{m+1}(x)-xA_{m-1}(x)^q\overset{\eqref{eq.defAA}}{=}A_m(x)^{q}+x^qA_{m-1}(x)^{q^2}-xA_{m-1}(x)^q\\
&\overset{\eqref{eq.defA}}{=}A_m(x)^{q}+x^qA_{m-1}(x)^{q^2}+x\left(x^{q}A_{m-2}(x)^{q^2}+A_m(x)\right)\\
&=x^q\left(A_{m-1}(x)+x^{q^{m-2}}A_{m-2}(x)\right)^{q^2}+A_m(x)^{q}+xA_m(x)\\
&=-x^qA_m(x)^{q^2}+A_m(x)^{q}+xA_m(x).
\end{aligned}\] \qed

\end{proof}

For $p$ even, we will further need the following proposition.
\begin{proposition}\label{Prop.Tr-Equalities}
 Let $p=2$. Let $a\in \GF{Q}$ with $G(a)\neq 0$.
 Let $E=\frac{aF(a)^{q+1}}{G^2(a)}$ and $H=\tr{d}\left(\frac{{Nr}^n_d(a)}{G^2(a)}\right)$. The
followings hold.
\begin{enumerate}
\item
\begin{equation}\label{tr_n}
\tr{n}(E)=mH.
\end{equation}
\item
\begin{equation}\label{tr_k}
T_k(E)=\frac{G(a)+F(a)^q}{G(a)}+\frac{k}{d}H.
\end{equation}
\end{enumerate}
\end{proposition}
\begin{proof}
Regarding the fact that ${Nr}^{mk}_{k}(a)={Nr}^n_d(a)$ as $a\in
\GF{Q}$, we have
\[\begin{aligned}
E&=\frac{aF(a)^{q+1}}{G(a)^2}\overset{\eqref{eq.Norm}}{=}\frac{aA_{m-1}(a)^qA_{m+1}+{Nr}^n_d(a)}{G(a)^2}=\frac{\left(A_{m+1}(a)+G(a)\right)A_{m+1}+{Nr}^n_d(a)}{G(a)^2}\\
&=\frac{A_{m+1}(a)}{G(a)}+\left(\frac{A_{m+1}(a)}{G(a)}\right)^2+\frac{{Nr}^n_d(a)}{G(a)^2}.\\
\end{aligned}\]
Hence, immediately \eqref{tr_n} follows from the facts
${Nr}^n_d(a)\in \GF{p^d}$ and $G(a)\in \GF{p^{md}}\cap
\GF{p^{k}}=\GF{p^d}$ (which follows from \eqref{p=2G(a)inGF(q)} as
$a\in \GF{p^{md}}$). And also
\[\begin{aligned}
T_k\left(E\right)&=\frac{A_{m+1}(a)}{G(a)}+\left(\frac{A_{m+1}(a)}{G(a)}\right)^q+\frac{k}{d}H\overset{\eqref{p=2G(a)inGF(q)}}{=}\frac{A_{m+1}(a)+A_{m+1}(a)^q}{G(a)}+\frac{k}{d}H\\
&=\frac{G(a)+aA_{m-1}(a)^q+A_{m+1}(a)^q}{G(a)}+\frac{k}{d}H\\
&\overset{\eqref{eq.defAA}}{=}\frac{G(a)+aA_{m-1}(a)^q+\left(A_{m}(a)+a^{q^{m-1}}A_{m-1}(a)\right)^q}{G(a)}+\frac{k}{d}H\\
&=\frac{G(a)+F(a)^q}{G(a)}+\frac{k}{d}H.
\end{aligned}\]\qed
\end{proof}

\section{Rational zeros of $P_a(X)$}

By exploiting the results of previous sections, now we can
completely solve the equation $P_a(X)=0$ in arbitrary finite fields.

\subsection{$N_a=p^d+1$}

 \begin{lemma}\label{Lem.p^d+1cond} Let $a\in \GF{Q}^*$. The following are equivalent.
 \begin{enumerate}
\item\label{1} $N_a=p^d+1$ i.e. $P_a(X)$ has exactly $p^d+1$ zeros in $\GF{Q}$.
\item\label{2} $F(a)=0$, or equivalently by Proposition \ref{Prop.A_r(x)=0}, there exists
$u\in \GF{q^m}\setminus\GF{q^2}$ such that
$a=\frac{(u-u^q)^{q^2+1}}{(u-u^{q^2})^{q+1}}$.
\item\label{3} There exists $u\in \GF{Q}\setminus\GF{p^{2d}}$ such that
$a=\frac{(u-u^q)^{q^2+1}}{(u-u^{q^2})^{q+1}}$. Then the $p^d+1$
zeros in $\GF{Q}$ of $P_a(X)$ are $x_0=\frac{-1}{1+(u-u^q)^{q-1}}$
and $x_\alpha=\frac{-(u+\alpha)^{q^2-q}}{1+(u-u^q)^{q-1}}$ for
$\alpha \in \GF{p^d}$.
\end{enumerate}
 \end{lemma}
\begin{proof}
(Item~\ref{1} $\Longleftrightarrow$ Item~\ref{2})

 We already showed that if $F(a)\neq 0$, then $N_a\leq 2$ i.e. $N_a\neq
 p^d+1$.

If $F(a)=0$ i.e.  there exists $u\in \GF{q^m}\setminus\GF{q^2}$ such
that $a=\frac{(u-u^q)^{q^2+1}}{(u-u^{q^2})^{q+1}}$, then  the set
given by
\[\bigcup_{\alpha \in \GF{q}} \{\frac{-(u+\alpha)^{q^2-q}}{1+(u-u^q)^{q-1}}\}\bigcup \{\frac{-1}{1+(u-u^q)^{q-1}}\}\]
is the set of all $q+1$ zeros of $P_{a}(X)$. In fact, the
cardinality of this set is exactly $q+1$ as $u$ is not in $\GF{q}$.
Also we have
\[\begin{aligned}
P_a\left(\frac{-1}{1+(u-u^q)^{q-1}}\right)&=\frac{-1}{1+(u-u^q)^{q-1}}\left(1-\frac{1}{1+(u-u^q)^{q-1}}\right)^q+\frac{(u-u^q)^{q^2+1}}{(u-u^{q^2})^{q+1}}\\
&=\frac{-(u-u^q)}{u-u^{q^2}}\left(\frac{(u-u^q)^q}{u-u^{q^2}}\right)^q+\frac{(u-u^q)^{q^2+1}}{(u-u^{q^2})^{q+1}}=0
\end{aligned}\]
and
\[\begin{aligned}
&P_a\left(\frac{-(u+\alpha)^{q^2-q}}{1+(u-u^q)^{q-1}}\right)=\frac{-(u+\alpha)^{q^2-q}}{1+(u-u^q)^{q-1}}\left(1+\frac{-(u+\alpha)^{q^2-q}}{1+(u-u^q)^{q-1}}\right)^q+\frac{(u-u^q)^{q^2+1}}{(u-u^{q^2})^{q+1}}\\
&=\frac{-(u-u^q)}{(u-u^{q^2})^{q+1}(u+\alpha)^{q}}\left((u-u^{q^2})(u+\alpha)^{q}-(u-u^q)(u+\alpha)^{q^2}\right)^q+\frac{(u-u^q)^{q^2+1}}{(u-u^{q^2})^{q+1}}\\
&=\frac{-(u-u^q)}{(u-u^{q^2})^{q+1}(u+\alpha)^q}\left((u-u^{q^2})(u^{q}+\alpha)-(u-u^q)(u^{q^2}+\alpha)\right)^q+\frac{(u-u^q)^{q^2+1}}{(u-u^{q^2})^{q+1}}\\
&=\frac{-(u-u^q)}{(u-u^{q^2})^{q+1}(u+\alpha)^q}\left((u-u^q)^q(u+\alpha)\right)^q+\frac{(u-u^q)^{q^2+1}}{(u-u^{q^2})^{q+1}}=0.\\
\end{aligned}\]
 Thus  $P_a(X)$
splits in $\GF{q^m}$. Corollary 7.2 of \cite{Bluher2004} states that
$P_a(X)$ splits in $\GF{q^m}$ if and only if $P_a(X)$ has exactly
$p^d+1$ zeros in $\GF{Q}$.

(Item~\ref{1} $\Longleftrightarrow$ Item~\ref{3})

To begin with, define $S_0=\GF{Q}\setminus \GF{q^2}$, $S_1=\{v\in
\GF{Q}\setminus \GF{q} \mid Tr_{d}^{n}(v)=0\}$, $S_2=\{v^{q-1}\mid
v\in S_1\}$ and $S=\{a\in \GF{Q} \mid N_a=p^d+1\}$.

 Now, we will show that the mapping
\[\Psi: u\in S_0\longmapsto
\frac{(u-u^q)^{q^2+1}}{(u-u^{q^2})^{q+1}}\in S,\] which is
well-defined by Proposition~\ref{Prop.A_r(x)=0} and the equivalence
between Item 1 and Item 2, is surjective.

Regarding
$\frac{(u-u^q)^{q^2+1}}{(u-u^{q^2})^{q+1}}=\frac{((u-u^q)^{q-1})^q}{(1+(u-u^q)^{q-1})^{q+1}}$,
we can write $\Psi=\varphi_3\circ\varphi_2\circ\varphi_1$ where
$\varphi_1: u\in S_0 \longmapsto u-u^q \in S_1$, $\varphi_2: v\in
S_1 \longmapsto v^{q-1} \in S_2$, $\varphi_3: w\in S_2 \longmapsto
\frac{w^q}{(1+w)^{q+1}}\in S$.

Consider $\varphi_1(u+\GF{p^d})=\varphi_1(u)$ for any $u\in S_0$ and
$\#S_1=p^{(m-1)d}-(p^d-(p^d-1)\cdot(m\mod 2))=(p^{md}-p^{(2-m
\mod2)d})/{p^d}=\#S_0/p^d$. Therefore $\varphi_1$ is $p^d$-to-one
and surjective. Next, consider that $\varphi_2(v_1)=\varphi_2(v_2)$
for $v_1, v_2\in \GF{Q}$ if and only if $v_2=\beta v_1$ for some
$\beta \in \GF{p^d}^*$ and that if $v_1\in S_1$ then $\beta v_1 \in
S_1$ for any $\beta \in \GF{p^d}^*$ since $Tr_{d}^{n}(\beta
v_1)=\beta Tr_{d}^{n}(v_1)=0$. Hence $\varphi_2$ is $(p^d-1)$-to-one
and surjective. On the other hand, if $a=\varphi_3(w)$ for $w\in
S_2$, then
$P_a(-\frac{1}{1+w})=\left(-\frac{1}{1+w}\right)^{q+1}+\left(-\frac{1}{1+w}\right)+\frac{w^q}{(1+w)^{q+1}}=0$.
Since $a\in S$ and so $N_a=p^d+1$, there are at most $p^d+1$  such
$w\in S_2$ that $\varphi_3(w)=a$. Therefore we get
\[ \#\Psi(S_0)\geq \frac{\#S_2}{p^d+1}=
\frac{p^{(m-1)d}-p^{(1-m \mod2)d}}{p^{2d}-1}.\] Since
$\#S=\frac{p^{(m-1)d}-p^{(1-m \mod2)d}}{p^{2d}-1}$ by Theorem 5.6 of
\cite{Bluher2004}, we have a sequence of inequalities $\#S\leq
\#\Psi(S_0) \leq \#S$ which concludes that $\Psi(S_0)=S$ i.e. $\Psi$
is
 surjective. (Note that it also follows that $\varphi_3$ is $(p^{d}+1)$-to-one and $\Psi$ is
 $p^d(p^{2d}-1)$-to-one.) This means that Item~\ref{1} and
 Item~\ref{3} are equivalent.
 \qed
\end{proof}

\subsection{$N_a\leq 2$: Odd $p$}

\begin{theorem}\label{Theo.oddp_N<3}
Let $p$ be odd. Let $a\in \GF{Q}$ and $E=G(a)^2-4a{F(a)}^{q+1}$.
\begin{enumerate}
\item $N_a=1$ if and only if $F(a)\neq 0$ and $E=0.$ In this case, the
unique zero in $\GF{Q}$ of $P_a(X)$ is $-\frac{G(a)}{2F(a)}.$
\item $N_a=0$ if and only if $E$ is not a quadratic residue in $\GF{p^d}$ (i.e.
$E^{\frac{p^d-1}{2}}\neq 0,1).$
\item $N_a=2$ if and only if  $E$ is a non-zero quadratic residue in $\GF{p^d}$ (i.e.
$E^{\frac{p^d-1}{2}}=1$). In this case, the two zeros in $\GF{Q}$ of
$P_a(X)$ are $x_{1,2}=\frac{\pm E^{\frac{1}{2}}-G(a)}{2F(a)}$, where
$E^{\frac{1}{2}}$ represents a quadratic root in $\GF{p^d}$ of $E$.
\end{enumerate}
\end{theorem}
\begin{proof}To begin with, note $E\in \GF{q}$ by
\eqref{property_FG_2} and so $E\in \GF{q}\cap \GF{Q}=\GF{p^d}$.

Assume $F(a)\neq 0.$ Then the equation \eqref{eq.quadEq} can be
rewritten as
\begin{equation}\label{quadeq_odd}
\left(x+\frac{G(a)}{2F(a)}\right)^2=\frac{E}{4{F(a)}^2}.
\end{equation}
Now, we will show that the solutions $x_{1,2}=\frac{\pm
E^{\frac{1}{2}}-G(a)}{2F(a)}$ of \eqref{quadeq_odd} become the zeros
of $P_a(X)$ if and only if $E$ is a quadratic residue in $\GF{q}$.
In fact, letting
$\left(E^{\frac{1}{2}}\right)^q=E^{\frac{1}{2}}+\delta$, we have
\[\begin{aligned}
P_a(x_{1,2})&=x_{1,2}(x_{1,2}+1)^q+a=\frac{\pm E^{\frac{1}{2}}-G(a)}{2F(a)}\left(1+\frac{\pm E^{\frac{1}{2}}-G(a)}{2F(a)}\right)^q+a\\
&=\frac{(\pm E^{\frac{1}{2}}-G(a))\left(\pm
E^{\frac{1}{2}}+\delta+(2F(a)-G(a))^q\right)+4a{F(a)}^{q+1}}{4{F(a)}^{q+1}}\\
&\overset{\eqref{property_FG_1}}{=}\frac{(\pm
E^{\frac{1}{2}}-G(a))\left(\pm
E^{\frac{1}{2}}+\delta+G(a)\right)+4a{F(a)}^{q+1}}{4{F(a)}^{q+1}}=\frac{(\pm
E^{\frac{1}{2}}-G(a))\delta}{4{F(a)}^{q+1}},
\end{aligned}\]
and so $P_a(x_{1,2})=0$ if and only if $\delta=0$, that is,
$E^{\frac{1}{2}}\in \GF{q}$. On the other hand, $x_{1,2}\in \GF{Q}$
if and only if $E^{\frac{1}{2}}\in \GF{Q}$. Combining above
discussion with Lemma \ref{Lem.p^d+1cond} completes the proof.\qed
\end{proof}
\begin{remark}
In the last two cases of Theorem~\ref{Theo.oddp_N<3} (i.e. the cases
of $N_a=0$ or $2$), the condition $F(a)\neq0$ is implied because
$E\neq 0$ implies $F(a)\neq0$. Indeed, if $F(a)=0$, then from
Equality~\eqref{property_FG_3} $G(a)=0$ follows and so $E=0$.
\end{remark}
\subsection{$N_a\leq 2$: $p=2$}

When $p=2$, Item 1 and 2 of Proposition~\ref{Prop.properties_FG} are
reduced to
\begin{equation}\label{p=2G(a)inGF(q)}
 G(x)\in \GF{q} \text{ for any } x\in \GF{q^m}.
\end{equation}

\begin{theorem}\label{Theo.p=2.N<3}
Let $p=2$ and $a\in \GF{Q}$. Let
$H=\tr{d}\left(\frac{{Nr}^n_d(a)}{G^2(a)}\right)$ and
$E=\frac{aF(a)^{q+1}}{G^2(a)}$.
\begin{enumerate}
\item $N_a=1$ if and only if $F(a)\neq 0$ and $G(a)=0$. In this case, $(aF(a)^{q-1})^{\frac{1}{2}}$ is the unique
zero in $\GF{Q}$ of $P_a(X)$.
\item $N_a=0$  if and only if  $G(a)\neq 0$ and $H\neq 0$.
\item $N_a=2$  if and only if  $G(a)\neq 0$ and $H=0$. In this
case the two zeros in $\GF{Q}$ are $x_1=\frac{G(a)}{F(a)}\cdot
T_n\left(\frac{E}{\zeta+1}\right)$ and $x_2=x_1+\frac{G(a)}{F(a)}$,
where $\zeta\in \mu_{Q+1}\setminus \{1\}$.
\end{enumerate}
\end{theorem}
\begin{proof} Let assume $F(a)\neq 0$. If $G(a)=0$, then the equation \eqref{eq.quadEq} has unique solution
$x_0=(aF(a)^{q-1})^{1/2}$. Then
$P_a^2(x_0)=\frac{a}{F(a)}\left(a^qF^{q^2}(a)+F^{q}(a)+aF(a)\right)\overset{\eqref{property_FG_3}}{=}\frac{a}{F(a)}G(a)=0$
 and thus it follows that $P_a(X)$ has exactly one zero
$(aF(a)^{q-1})^{1/2}$ in $\GF{Q}$ when $G(a)=0$.

Now consider the case of $G(a)\neq 0$. Note that
\eqref{property_FG_3} shows that $G(a)\neq 0$ implies $F(a)\neq 0$.
The equation \eqref{eq.quadEq} can be rewritten as
$\left(\frac{F(a)}{G(a)}x\right)^2+\frac{F(a)}{G(a)}x=E$ and so it
has a solution in $\GF{Q}$  if and only if
\begin{equation}\label{tr_n=0}
\tr{n}\left(E\right)=0.
\end{equation}
In case of existence of solution, \eqref{eq.quadEq} has exactly two
solutions $x_1$ and $x_2$ in $\GF{Q}$. Indeed,
$\left(\frac{F(a)}{G(a)}x_1\right)^2+\frac{F(a)}{G(a)}x_1=\left(\frac{F(a)}{G(a)}x_2\right)^2+\frac{F(a)}{G(a)}x_2=
T_n\left(\frac{E}{\zeta+1}\right)^2+T_n\left(\frac{E}{\zeta+1}\right)=\left(\frac{E}{\zeta+1}\right)^{Q}+\left(\frac{E}{\zeta+1}\right)
=E(\frac{1}{\frac{1}{\zeta}+1}+\frac{1}{\zeta+1})=E$, and so $x_1$
and $x_2$ are two solutions of \eqref{eq.quadEq}. And, both $x_1$
and  $x_2$ are in $\GF{Q}$ since
$T_n\left(\frac{E}{\zeta+1}\right)^Q+T_n\left(\frac{E}{\zeta+1}\right)=T_n(E)=\tr{n}\left(E\right)\overset{\eqref{tr_n=0}}{=}0$
i.e. $T_n\left(\frac{E}{\zeta+1}\right)\in \GF{Q}$.

Let $x$ be a solution of \eqref{eq.quadEq}.  Then we have
$T_k\left(\left(\frac{F(a)}{G(a)}x\right)^2+\frac{F(a)}{G(a)}x\right)=T_k(E)$,
i.e. $\left(\frac{F(a)}{G(a)}x\right)^q+\frac{F(a)}{G(a)}x=T_k(E)$,
hence
$x^q=\left(\frac{G(a)}{F(a)}\right)^q\left(\frac{F(a)}{G(a)}x+T_k(E)\right)\overset{\eqref{p=2G(a)inGF(q)}}{=}\frac{F(a)x+G(a)T_k(E)}{F(a)^q}$
and
$P_a(x)=x(x^q+1)+a=\frac{x\left(F(a)x+G(a)T_k(E)+F(a)^q\right)}{F(a)^q}+a\overset{\eqref{eq.quadEq}}{=}\frac{x\left(G(a)T_k(E)+F(a)^q+G(a)\right)}{F(a)^q}$.

Thus, it follows that the solution $x$ of  \eqref{eq.quadEq} is zero
of $P_a(X)$ if and only if
\begin{equation}\label{tr_k=0}
T_k\left(E\right)=\frac{G(a)+F(a)^q}{G(a)}.
\end{equation}
Equalities \eqref{tr_n}, \eqref{tr_k}, \eqref{tr_n=0} and
\eqref{tr_k=0} together leads us to conclude that when $G(a)\neq 0$,
$P_a(X)$ has a zero (equivalently, exactly two zeros) in $\GF{Q}$ if
and only if $mH=\frac{k}{d}H=0$ which is equivalent to $H=0$ since
at least one of $m$ and $k/d$ must be odd as $\gcd(m, k/d)=1$.

Combining above discussion with Lemma \ref{Lem.p^d+1cond} completes
the proof.
 \qed
\end{proof}

\begin{remark}
When $p=2$,  $A_r(X)$ defined in this paper coincides with $C_r(X)$
introduced in \cite{HK2010}. Many of our results specific to $p=2$
also appears in \cite{HK2010} with relatively longer and complicate
proof.
\end{remark}

\begin{remark}
On the other hand, very recently,  the number of  roots of
linearized and projective polynomials was studied in
\cite{CMPZ2019,MS2019}. In particular, criteria for which $P_a(X)$
has $0$, $1$, $2$ or $p^d+1$ roots were stated by Theorem 8 of
\cite{MS2019} using some polynomial sequence $G_r(X)$ which are
related by $A_r(X)=G_{r-1}(X)^q$ with $A_r(X)$ defined in this
paper. Using the notations of our paper, Theorem 8 of \cite{MS2019}
states that $N_a=p^d+1$ if and only if $A_m(a)=0$ and $A_{m+1}(a)\in
\GF{p^d}$. As the first note, here, the condition $A_{m+1}(a)\in
\GF{p^d}$ is surplus because this follows from the condition
$A_m(a)=0$. In fact, if $F(a)=A_m(a)=0$ then by \eqref{eq.defA}
$A_{m+1}(a)=(-aA_{m-1}(a)^q)^q$  and by \eqref{property_FG_3}
$G(a)=0$ i.e. $A_{m+1}(a)=-aA_{m-1}(a)^q$, so
$A_{m+1}(a)=A_{m+1}(a)^q$ that is $A_{m+1}(a)\in \GF{q}\cap
\GF{Q}=\GF{p^d}$.

As the second note, when $p=2$, the criteria for $N_a=0,1,2$ in
\cite{MS2019} are false. In the criteria for $N_a=0,1,2$ of Theorem
8 of \cite{MS2019}, $G_n\in \GF{q}$ or $G_n\notin \GF{q}$ must be
fixed by $G_n+{G_n}^\sigma+{G_{n-1}}^\sigma=0$ or
$G_n+{G_n}^\sigma+{G_{n-1}}^\sigma\neq 0$ respectively. Note that
the quantity $G_n+{G_n}^\sigma+{G_{n-1}}^\sigma$ ($\Delta_L$ for $p$
odd, resp.) therein equals $G(a)^{\frac{1}{q}}$ ($E^{\frac{1}{q}}$
for $p$ odd, resp.) in the notation of our paper.
\end{remark}

\section{More for the case $N_a=p^d+1$}
Let $S_a=\{x\in \GF{p^{md}}=\GF{Q} \mid P_a(x)=0\}$. The following problem is remained: when $N_a=p^d+1$ i.e. $A_m(a)=0$, express $S_a$ explicitly in terms of $a$.\\
For this problem, the following facts are the only things we know at
the moment.
\begin{enumerate}
\item When $m=3$ and $A_3(a)=1-a^q=0$ i.e. $a=1$, we have
\[S_a=\{(b-b^q)^{q-1}, b\in
\GF{p^{3d}}\setminus \GF{p^d}\}.\]
\item When $p=2$, $m=4$ and $A_4(a)=1+a^q+a^{q^2}=0$, we have \[\sqrt{a}\in S_{a}.\]
\item When $p=2$, $m=5$ and $A_5(a)=1+a^q+a^{q^2}+a^{q^3}(1+a^q)=0$, we have \[\frac{a(a+a^q)}{1+a^q+a^{q+1}}\in S_{a}.\]
\item When $p=2$, $m=6$ and $A_6(a)=1+a^q+a^{q^2}+a^{q^3}(1+a^q)+a^{q^4}(1+a^q+a^{q^2})=0$, we have \[\sqrt{\frac{a^2(1+a+a^q+a^{q^2+1})+a^{q^2+q+1}(1+a+a^q)^q}{a^{2q^2+q}+(1+a+a^q)(1+a^2+a^q)^q}}\in S_{a}.\]
\end{enumerate}
All these can be checked by direct substitutions to $P_a(X)$.

\begin{lemma}
If $x^{q+1}+x+a=0$ for $a\in \GF{Q}^*$, then for any $r\geq 0$
\begin{equation}\label{eq:x^{q^r}}
x^{q^r}=\frac{A_{r+1}(a)x-aA_{r}(a)^q}{A_r(a)x-aA_{r-1}(a)^q},
\end{equation}
where the denominator never equal zero.
\end{lemma}
\begin{proof}
This is an alternation of \eqref{eq.xqr}. Only thing to be verified
is the fact that the denominator never equal zero. In fact, if
$A_r(a)x-aA_{r-1}(a)^q=0$ (and so also $A_{r+1}(a)x-aA_{r}(a)^q=0$
by \eqref{eq.xqr}), then
$x=\frac{aA_{r-1}(a)^q}{A_r(a)}=\frac{aA_{r}(a)^q}{A_{r+1}(a)}$ and
thus it follows $a(A_{r}(a)^{q+1}-A_{r-1}(a)^qA_{r+1}(a))=0$. But
\eqref{eq.Norm} shows
$a(A_{r}(a)^{q+1}-A_{r-1}(a)^qA_{r+1}(a))=a^{\frac{q^r-1}{q-1}}\neq
0$ i.e. a contradiction.\qed
\end{proof}

\begin{lemma}
If $A_{m}(a)=0$, then  for any $x\in \GF{Q}$ such that
$x^{q+1}+x+a=0$, it holds
\[
Nr_{k}^{km}(x)=A_{m+1}(a).
\]
Furthermore, for any $t\geq 0$
\[
A_{m+t}(a)=A_{m+1}(a)\cdot A_{t}(a).
\]
\end{lemma}
\begin{proof}
By multiplying all equalities \eqref{eq:x^{q^r}} for $r$ ranging
from 1 to $m-1$
 side by side we get
 $x^{\frac{q^m-1}{q-1}}=-aA_{m-1}(a)^q=A_{m+1}(a)^{1/q}$, i.e. $A_{m+1}(a)=Nr_{k}^{km}(x)^q=Nr_{k}^{km}(x) \in
 \GF{q}$.
 Then, an induction on $t$ leads to the conclusion of the lemma.\qed
\end{proof}

\section{Conclusions}
We studied the equation $P_a(X)=X^{p^k+1}+X+a=0, a\in \GF{p^n}$ and
proved some new criteria for the number of the $\GF{p^n}$-zeros of
$P_a(x)$. For the cases of one or two $\GF{p^n}$-zeros, we provided
explicit expressions for these rational zeros in terms of $a$. For
the case of $p^{\gcd(n, k)}+1$ rational zeros, we provided a
parametrization of such $a$'s and expressed all the $p^{\gcd(n,
k)}+1$ rational zeros by using this parametrization.   An important
remaining problem is whether for any given $p,n,k$, in the case of
$p^{\gcd(n, k)}+1$ rational zeros, it is always possible to
explicitly express these $p^{\gcd(n, k)}+1$ rational zeros in terms
of $a$.

\end{document}